\documentclass[a4paper,UKenglish]{lipics-v2018}

\usepackage{microtype}
\usepackage[noend]{algpseudocode}
\usepackage{xspace,xcolor}
\usepackage{algorithm}
\usepackage{float}
\usepackage{subcaption}
\usepackage{complexity}
\usepackage{gastex}
\usepackage{tikz}
\usepackage{booktabs}
\usepackage{wrapfig}

\hideLIPIcs

\nolinenumbers

\title{Minimum Segmentation for Pan-genomic Founder~Reconstruction in Linear Time}

\titlerunning{Minimum Segmentation for Pan-genomic Founder Reconstruction in Linear Time}

\author{Tuukka Norri}{Department of Computer Science, University of Helsinki, Helsinki, Finland}{tuukka.norri@helsinki.fi}{https://orcid.org/0000-0002-8276-0585}{}

\author{Bastien Cazaux}{Department of Computer Science, University of Helsinki, Helsinki, Finland}{bastien.cazaux@helsinki.fi}{https://orcid.org/0000-0002-1761-4354}{}

\author{Dmitry Kosolobov}{Department of Computer Science, University of Helsinki, Helsinki, Finland}{dkosolobov@mail.ru}{https://orcid.org/0000-0002-2909-2952}{}

\author{Veli M\"akinen}{Department of Computer Science, University of Helsinki, Helsinki, Finland}{veli.makinen@helsinki.fi}{https://orcid.org/0000-0003-4454-1493}{}

\authorrunning{T. Norri, B. Cazaux, D. Kosolobov, V. M\"akinen}

\Copyright{CC-BY}

\subjclass{}

\category{}

\relatedversion{This is a preprint of a paper in WABI 2018 \cite[\url{https://doi.org/10.4230/LIPIcs.WABI.2018.15}]{NCKM18}.}

\supplement{}

\funding{}

\acknowledgements{}

\EventEditors{John Q. Open and Joan R. Access}
\EventNoEds{2}
\EventLongTitle{42nd Conference on Very Important Topics (CVIT 2016)}
\EventShortTitle{CVIT 2016}
\EventAcronym{CVIT}
\EventYear{2016}
\EventDate{December 24--27, 2016}
\EventLocation{Little Whinging, United Kingdom}
\EventLogo{}
\SeriesVolume{42}
\ArticleNo{23}

\newcommand{\card}[1]{\vert #1 \vert}

\newcommand{\acro}[1]{\ensuremath{\mathtt{#1}}\xspace}

\newcommand{\lcp}{\acro{LCP}}

\begin{document}
\maketitle


\begin{abstract}
Given a threshold $L$ and a set $\mathcal{R} = \{R_1, \ldots, R_m\}$ of $m$ haplotype sequences, each having length $n$, the minimum segmentation problem for founder reconstruction is to partition the sequences into disjoint segments $\mathcal{R}[i_1{+}1,i_2], \mathcal{R}[i_2{+}1, i_3], \ldots, \mathcal{R}[i_{r-1}{+}1, i_r]$, where $0 = i_1 < \cdots < i_r = n$ and $\mathcal{R}[i_{j-1}{+}1, i_j]$ is the set $\{R_1[i_{j-1}{+}1, i_j], \ldots, R_m[i_{j-1}{+}1, i_j]\}$, such that the length of each segment, $i_j - i_{j-1}$, is at least $L$ and $K = \max_j\{ |\mathcal{R}[i_{j-1}{+}1, i_j]| \}$ is minimized. The distinct substrings in the segments $\mathcal{R}[i_{j-1}{+}1, i_j]$ represent founder blocks that can be concatenated to form $K$ founder sequences representing the original $\mathcal{R}$ such that crossovers happen only at segment boundaries. We give an optimal $O(mn)$ time algorithm to solve the problem, improving over earlier $O(mn^2)$. This improvement enables to exploit the algorithm on a pan-genomic setting of haplotypes being complete human chromosomes, with a goal of finding a representative set of references that can be indexed for read alignment and variant calling.
\end{abstract}

\section{Introduction}
A key problem in \emph{pan-genomics} is to develop a sufficiently small, efficiently queriable, but still descriptive representation of the variation common to the subject under study \cite{PGWhitePaper}. For example, when studying human population, one would like to take all publicly available variation datasets (e.g. \cite{Nat15,EXAC,UK10K}) into account. Many approaches encode the variation as a graph \cite{Sch09,HPB13,SVM14,Dil15,Iqp16,vg} and then one can encode the different haplotypes as paths in this graph. An alternative was proposed in \cite{VNVPM18}, based on a compressed indexing scheme for a multiple alignment of all the haplotypes \cite{MNSV09jcb,N12,WSBS13,FGHP14,GP15}. In either approach, scalability is hampered by the encoding of all the haplotypes.

We suggest to look for a smaller set of representative haplotype sequences to make the above pan-genomic representations scalable.

Finding such set of representative haplotype sequences that retain the original contiguities as well as possible, is known as the \emph{founder sequence reconstruction} problem \cite{Ukkonen02}. In this problem, one seeks a set of $k$ founders such that the original $m$ haplotypes can be mapped with minimum amount of \emph{crossovers} to the founders. Here a crossover means a position where one needs to jump from one founder to another to continue matching the content of the haplotype in question. Unfortunately, this problem in $\NP$-hard even to approximate within a constant factor \cite{RastasU07}.

For founder reconstruction to be scalable to the pan-genomic setting, one would need an algorithm to be nearly linear to the input size. There is only one relaxation of founder reconstruction that is polynomial time solvable. Namely, when limiting all the crossovers to happen at the same locations, one obtains a \emph{minimum segmentation} problem specific to founder reconstruction~\cite{Ukkonen02}. A dynamic programming algorithm given in~\cite{Ukkonen02} has complexity $O(n^2m)$, where $m$ is the number of haplotypes and $n$ is the length of each of them.

In this paper, we improve the running time of solving the minimum segmentation problem of founder reconstruction to the optimal $O(mn)$ (linear in the input size).

The main technique behind the improvement is the use of \emph{positional Burrows--Wheeler transform} (pBWT)~\cite{Durbin14}, or more specificly its extension to larger alphabets \cite{MakinenNorri}. While the original dynamic programming solution uses $O(nm)$ time to look for the best preceding segment boundary for each column of the input, we observe that at most $m$ values in pBWT determine segment boundaries where the number of distinct founder substrings change. Minimums on the already computed dynamic programming values between each such interesting consecutive segment boundaries give the requested result. However, it turns that we can maintain the minimums directly in pBWT internal structures (with some modifications) and have to store only the last $L$ computed dynamic programming values, thus spending only $O(m + L)$ additional space, where $L$ is the input threshold on the length of each segment. The segmentation is then reconstructed by standard backtracking approach in $O(n)$ time using an array of length $n$.

\section{Notation and Problem Statement}

For a string $s = c_1 c_2 \cdots c_n$, denote by $|s|$ its length $n$. We write $s[i]$ for the letter $c_i$ of $s$ and $s[i,j]$ for the \emph{substring} $c_i c_{i + 1} \cdots c_j$. An analogous notation is used for arrays. For any numbers $i$ and $j$, the set of integers $\{x \in \mathbb{Z} \colon i \le x \le j\}$ (possibly empty) is denoted by $[i,j]$.

The input for our problem is the set $\mathcal{R} = \{R_1,\ldots,R_m\}$ of strings of length $n$, called \emph{recombinants}. A set $\mathcal{F} = \{F_1,\ldots,F_d\}$ of strings of length $n$ is called a \emph{founder set} of $\mathcal{R}$ if for each string $R_i \in \mathcal{R}$, there exists a sequence $P_i$ of length $n$ such that, for all $j \in [1,n]$, we have $P_i[j] \in [1,d]$ and $R_i[j] = F_{P_i[j]}[j]$. The sequence $P_i$ is called a \emph{parse of $R_i$ in terms of $\mathcal{F}$} and the set of parses $\{P_1, \ldots, P_m\}$ is called a \emph{parse of $\mathcal{R}$ in terms of $\mathcal{F}$}. An integer $j$ such that $P_i[j-1] \neq P_i[j]$ is called a \emph{crossover point} of the parse $P_i$; for technical reasons, the integers $1$ and $n + 1$ are called crossover points too.

We consider the problem of finding a ``good'' founder set $\mathcal{F}$ and a ``good'' corresponding parse of $\mathcal{R}$ according to a reasonable measure of goodness. Ukkonen~\cite{Ukkonen02} pointed out that such measures may contradict each other: for instance, a minimum founder set obviously has size $d = \max_{j\in [1,n]} \card{\{R_1[j], \ldots, R_m[j]\}}$, but parses corresponding to such set may have unnaturally many crossover points; conversely, $\mathcal{R}$ is a founder set of itself and the only crossover points of its trivial parse are $1$ and $n + 1$, but the size $m$ of this founder set is in most cases unacceptably large. Following Ukkonen's approach, we consider compromise parameterized solutions. The \emph{minimum founder set problem}~\cite{Ukkonen02} is, given a bound $L$ and a set of recombinants $\mathcal{R}$, to find a smallest founder set $\mathcal{F}$ of $\mathcal{R}$ such that there exists a parse of $\mathcal{R}$ in terms of $\mathcal{F}$ in which the distance between any two crossover points is at least $L$.

It is convenient to reformulate the problem in terms of segmentations of $\mathcal{R}$. A \emph{segment} of $\mathcal{R} = \{R_1,\ldots,R_m\}$ is a set $\mathcal{R}[j,k] = \{R_i[j,k] \colon R_i \in \mathcal{R}\}$. A \emph{segmentation} of $\mathcal{R}$ is a collection $S$ of disjoint segments that covers the whole $\mathcal{R}$, i.e., for any distinct $\mathcal{R}[j,k]$ and $\mathcal{R}[j',k']$ from $S$, $[j,k]$ and $[j',k']$ do not intersect and, for each $x \in [1,n]$, there is $\mathcal{R}[j,k]$ from $S$ such that $x \in [j,k]$. The \emph{minimum segmentation problem} is, given a bound $L$ and a set of recombinants $\mathcal{R}$, to find a segmentation $S$ of $\mathcal{R}$ such that $\max\{\card{\mathcal{R}[j,k]} \colon \mathcal{R}[j,k] \in S\}$ is minimized and the length of each segment from $S$ is at least $L$; in other words, the problem is to compute
\begin{equation}\label{eq:min:segmentation:problem}
\min\limits_{S\in S_L} \max \{\card{\mathcal{R}[j,k]} \colon \mathcal{R}[j,k] \in S \},
\end{equation}
where $S_L$ is the set of all segmentations in which all segments have length at least $L$.

The minimum founder set problem and the minimum segmentation problem are, in a sense, equivalent: any segmentation $S$ with segments of length at least $L$ induces in an obvious way a founder set of size $\max \{ \card{\mathcal{R}[j,k]} \colon \mathcal{R}[j,k] \in S \}$ and a parse in which all crossover points are located at segment boundaries (and, hence, at distance at least $L$ from each other); conversely, if $\mathcal{F}$ is a founder set of $\mathcal{R}$ and $\{j_1, \ldots, j_p\}$ is the sorted set of all crossover points in a parse of $\mathcal{R}$ such that $j_q - j_{q-1} \ge L$ for $q \in [2,p]$, then $S = \{\mathcal{R}[j_{q-1},j_q{-}1] \colon q \in [2,p]\}$ is a segmentation of $\mathcal{R}$ with segments of length at least $L$ and $\max \{ \card{\mathcal{R}[j,k]} \colon \mathcal{R}[j,k] \in S \} \le |\mathcal{F}|$.

Our main result is an algorithm that solves the minimum segmentation problem in the optimal $O(mn)$ time. The solution normally does not uniquely define a founder set of $\mathcal{R}$: for instance, if the built segmentation of $\mathcal{R} = \{baaaa, baaab, babab\}$ is $S = \{\mathcal{R}[1,1], \mathcal{R}[2,3], \mathcal{R}[4,5]\}$, then the possible founder sets induced by $S$ are $\mathcal{F}_1 = \{baaab, babaa\}$ and $\mathcal{F}_2 = \{baaaa, babab\}$. In other words, to construct a founder set, one concatenates fragments of recombinants corresponding to the found segments in a certain order. One can use heuristics aiming to minimize the number of crossover points in founder set constructed in such a way \cite{Ukkonen02}. Our techniques extend to implementing such heuristics fast, but we leave the details for later and focus here on the segmentation problem.

Hereafter, we assume that the input alphabet $\Sigma$ is the set $[0..|\Sigma|{-}1]$ of size $O(m)$, which is a natural assumption considering that the typical alphabet size is 4 in our problem. It is sometimes convenient to view the set $\mathcal{R} = \{R_1, \ldots, R_m\}$ as a matrix with $m$ rows and $n$ columns. We say that an algorithm processing the recombinants $\mathcal{R}$ is streaming if it reads the input from left to right ``columnwise'', for each $k$ from 1 to $n$, and outputs an answer for each set of recombinants $\{R_1[1,k], \ldots, R_m[1,k]\}$ immediately after reading the ``column'' $\{R_1[k], \ldots, R_m[k]\}$. The main result of the paper is the following theorem.

\begin{theorem}\label{thm:maintheorem}
Given a bound $L$ and recombinants $\mathcal{R} = \{R_1, \ldots, R_m\}$, each having length $n$, there is an algorithm that computes~\eqref{eq:min:segmentation:problem} in a streaming fashion in the optimal $O(mn)$ time and $O(m + L)$ space; using an additional array of length $n$, one can also find in $O(n)$ time a segmentation on which \eqref{eq:min:segmentation:problem}~is attained, thus solving the minimum segmentation problem.
\end{theorem}

\section{Minimum Segmentation Problem}

Given a bound $L$ and a set of recombinants $\mathcal{R} = \{R_1, \ldots, R_m\}$ each of which has length~$n$, Ukkonen~\cite{Ukkonen02} proposed a dynamic programming algorithm that solves the minimum segmentation problem in $O(m n^2)$ time based on the following recurrence relation:
\begin{equation}\label{eq:ukkonen}
M(k) =
\begin{cases}
+\infty & \text{ if }k < L,\\
\card{\mathcal{R}[1,k]} & \text{ if }L \le k < 2L,\\
\min\limits_{0 \le j \le k-L} \max\{M(j),\card{\mathcal{R}[j+1,k]}\} & \text{ if } k \ge 2L.
\end{cases}
\end{equation}
It is obvious that $M(n)$ is equal to the solution \eqref{eq:min:segmentation:problem}; the segmentation itself can be reconstructed by ``backtracking'' in a standard way (see~\cite{Ukkonen02}). We build on the same approach.

For a given $k \in [1,n]$, denote by $j_{k,1}, \ldots, j_{k,r_k}$ the sequence of all positions $j \in [1, k - L]$ in which the value of $|\mathcal{R}[j,k]|$ changes, i.e., $1 \le j_{k,1} < \cdots < j_{k,r_k} \le k - L$ and $|\mathcal{R}[j_{k,h},k]| \neq |\mathcal{R}[j_{k,h}{+}1,k]|$ for $h \in [1,r_k]$. We complement this sequence with $j_{k,0} = 0$ and $j_{k,r_k+1} = k - L + 1$, so that $j_{k,0}, \ldots, j_{k,r_k+1}$ can be interpreted as a splitting of the range $[0,k - L]$ into segments in which the value $\card{\mathcal{R}[j+1,k]}$ stays the same: namely, for $h \in [0,r_k]$, one has $\card{\mathcal{R}[j + 1,k]} = \card{\mathcal{R}[j_{k,h+1},k]}$ provided $j_{k,h} \le j < j_{k,h+1}$. Hence, $\min\limits_{j_{k,h} \le j < j_{k,h+1}} \max\{M(j), \card{\mathcal{R}[j + 1,k]}\} = \max\{\card{\mathcal{R}[j_{k,h+1},k]}, \min\limits_{j_{k,h} \le j < j_{k,h+1}} M(j)\}$ and, therefore, \eqref{eq:ukkonen} can be rewritten as follows:
\begin{equation}\label{eq:mainrelation}
M(k) =
\begin{cases}
+\infty & \text{ if }k < L,\\
\card{\mathcal{R}[1,k]} & \text{ if } L \le k < 2L,\\
\min\limits_{0 \le h \le r_k} \max\{\card{\mathcal{R}[j_{k,h+1},k]}, \min\limits_{j_{k,h} \le j < j_{k,h+1}} M(j)\} & \text{ if } k \ge 2L.
\end{cases}
\end{equation}

Our crucial observation is that, for $k \in [1,n]$ and $j \in [1,k]$, one has $\card{\mathcal{R}[j + 1, k]} \le \card{\mathcal{R}[j, k]} \le m$.
Therefore, $m \ge \card{\mathcal{R}[j_{k,1},k]} > \cdots > \card{\mathcal{R}[j_{k,r_k+1},k]} \ge 1$ and $r_k < m$. Hence, $M(k)$ can be computed in $O(m)$ time using~\eqref{eq:mainrelation}, provided one has the following components:
\begin{enumerate}[(i)]
\item the numbers $\card{\mathcal{R}[j_{k,h+1},k]}$, for $h \in [0,r_k]$;
\item the values $\min\{ M(j) \colon j_{k,h} \le j < j_{k,h+1} \}$, for $h \in [0,r_k]$.
\end{enumerate}
In the remaining part of the section, we describe a streaming algorithm that reads the strings $\{R_1, \ldots, R_m\}$ ``columnwise'' from left to right and computes the components (i)~and~(ii) immediately after reading each ``column'' $\{R_1[k], \ldots, R_m[k]\}$, for $k \in [1,n]$, and all in $O(mn)$ total time and $O(m + L)$ space.

To reconstruct a segmentation corresponding to the found solution $M(n)$, we build along with the values $M(k)$ an array of size $n$ whose $k$th element, for each $k \in [1,n]$, stores $0$ if $M(k) = \card{\mathcal{R}[1,k]}$, and stores a number $j \in [1,k{-}L]$ such that $M(k) = \max\{M(j), \card{\mathcal{R}[j{+}1,k]}\}$ otherwise; then, the segmentation can be reconstructed from the array in an obvious way in $O(n)$ time. In order to maintain the array, our algorithm computes, for each $k \in [1,n]$, along with the values $\min\{ M(j) \colon j_{k,h} \le j < j_{k,h+1} \}$, for $h \in [0,r_k]$, positions $j$ on which these minima are attained (see below). Further details are straightforward and, thence, omitted.

\subsection{Positional Burrows--Wheeler Transform}\label{subsec:pbwt}

Let us fix $k \in [1,n]$. Throughout this subsection, the string $R_i[k] R_i[k-1] \cdots R_i[1]$, which is the reversal of $R_i[1,k]$, is denoted by $R'_{i,k}$, for $i \in [1,m]$. Given a set of recombinants $\mathcal{R} = \{R_1, \ldots, R_m\}$ each of which has length $n$, a \emph{positional Burrows--Wheeler transform (pBWT)}, as defined by Durbin~\cite{Durbin14}, is a pair of integer arrays $a_k[1,m]$ and $d_k[1,m]$ such that:
\begin{enumerate}
\item $a_k[1,m]$ is a permutation of $[1,m]$ such that $R'_{a_k[1],k} \le \cdots \le R'_{a_k[m],k}$ lexicographically;
\item $d_k[i]$, for $i \in [1,m]$, is an integer such that $R_{a_k[i]}[d_k[i] .. k]$ is the longest common suffix of $R_{a_k[i]}[1,k]$ and $R_{a_k[i-1]}[1,k]$, and $d_k[i] = k + 1$ if either this suffix is empty or $i = 1$.
\end{enumerate}

\begin{example}\label{firstexample}
Consider the following example, where $m = 6$, $k = 7$, and $\Sigma = \{a,c,t\}$. It is easy to see that the pBWT implicitly encodes the trie depicted in the right part of Figure~\ref{fig:pbwt:example}, and such interpretation drives the intuition behind this structure.
\begin{figure}[H]
\centering
\begin{subtable}[t]{0.12\textwidth}
\vskip-30mm
\begin{tabular}{l}
$R_1 = tttccat$\\
$R_2 = accatta$\\
$R_3 = actacct$\\
$R_4 = actccat$\\
$R_5 = cttacct$\\
$R_6 = atcacat$
\end{tabular}
\end{subtable}
\hfill
\begin{subtable}[t]{0.38\textwidth}
\centering
\vskip-30mm
\small
\begin{tabular}{r|cccccc}
$i$      & $1$ & $2$ & $3$ & $4$ & $5$ & $6$\\\midrule
$a_k[i]$ & $2$ & $6$ & $4$ & $1$ & $3$ & $5$\\
$d_k[i]$ & $8$ & $8$ & $5$ & $3$ & $7$ & $3$\\\bottomrule
\end{tabular}
\vskip4mm
\begin{tabular}{r|ccccccc}
$i$                       & $1$ & $2$ & $3$ & $4$ & $5$ & $6$ & $7$\\\midrule
$\card{\mathcal{R}[i,k]}$ & $6$ & $6$ & $4$ & $4$ & $3$ & $3$ & $2$\\\bottomrule
\end{tabular}
\end{subtable}
\hfill
\begin{subfigure}[t]{0.4\textwidth}
\begin{tikzpicture}
\tikzstyle{every node}=[anchor=base,minimum width=0.7cm,inner sep=1pt]
\draw[step=.5, gray, very thin] (0.5,0) grid (4.0,3.0);
\node at (-0.5,0.1) (m1) {$a_k[6] = 5$};
\node at (-0.5,0.6) (m2) {$a_k[5] = 3$};
\node at (-0.5,1.1) (m3) {$a_k[4] = 1$};
\node at (-0.5,1.6) (m4) {$a_k[3] = 4$};
\node at (-0.5,2.1) (m5) {$a_k[2] = 6$};
\node at (-0.5,2.6) (m6) {$a_k[1] = 2$};
\draw[cyan,line width=4pt] (0.5,0) to (1.5,0) to (1.5,0.5);
\draw[cyan,line width=4pt] (0.5,0.5) to (3.5,0.5) to (3.5,1);
\draw[cyan,line width=4pt] (0.5,1.5) to (1.5,1.5) to (1.5,1);
\draw[cyan,line width=4pt] (0.5,2) to (2.5,2) to (2.5,1);
\draw[cyan,line width=4pt] (0.5,2.5) to (4,2.5) to (4,1) to (0.5,1);
\node at (0.75,0.1) () {$c$};\node at (1.25,0.1) () {$t$};
\node at (0.75,0.6) () {$a$};\node at (1.25,0.6) () {$c$};\node at (1.75,0.6) () {$t$};\node at (2.25,0.6) () {$a$};\node at (2.75,0.6) () {$c$};\node at (3.25,0.6) () {$c$};
\node at (0.75,1.1) () {$t$};\node at (1.25,1.1) () {$t$};\node at (1.75,1.1) () {$t$};\node at (2.25,1.1) () {$c$};\node at (2.75,1.1) () {$c$};\node at (3.25,1.1) () {$a$};\node at (3.75,1.1) () {$t$};
\node at (0.75,1.6) () {$a$};\node at (1.25,1.6) () {$c$};
\node at (0.75,2.1) () {$a$};\node at (1.25,2.1) () {$t$};\node at (1.75,2.1) () {$c$};\node at (2.25,2.1) () {$a$};
\node at (0.75,2.6) () {$a$};\node at (1.25,2.6) () {$c$};\node at (1.75,2.6) () {$c$};\node at (2.25,2.6) () {$a$};\node at (2.75,2.6) () {$t$};\node at (3.25,2.6) () {$t$};\node at (3.75,2.6) () {$a$};
\end{tikzpicture}
\end{subfigure}
\caption{The pBWT for a set of recombinants $\mathcal{R} = \{R_1, \ldots, R_6\}$ and some additional information.}\label{fig:pbwt:example}
\end{figure}
\end{example}

Durbin~\cite{Durbin14} showed that $a_k$ and $d_k$ can be computed from $a_{k-1}$ and $d_{k-1}$ in $O(m)$ time on the binary alphabet. M\"akinen and Norri~\cite{MakinenNorri} further generalized the construction for integer alphabets of size $O(m)$, as in our case. For the sake of completeness, we describe in this subsection the solution from~\cite{MakinenNorri} (see Figure~\ref{fig:pbwt:basic}), which serves then as a basis for our main algorithm. We also present a modification of this solution (see Figure~\ref{fig:pbwt:modified}), which, albeit seems to be slightly inferior in theory (we could prove only $O(m\log|\Sigma|)$ time upper bound), showed better performance in practice and thus, as we believe, is interesting by itself.

\begin{figure}[htb]
\centering
\begin{subfigure}[t]{.48\textwidth}
\begin{algorithmic}
\small
	\State zero initialize $C[0,|\Sigma|]$ and $P[0, |\Sigma| - 1]$
	\For{$i \gets 1 \mathrel{\mathbf{to}} m$}
        \State $C[R_i[k] + 1] \gets C[R_i[k] + 1] + 1$;
    \EndFor
	\For{$i \gets 1 \mathrel{\mathbf{to}} |\Sigma|-1$}
        $C[i] \gets C[i] + C[i - 1]$;
    \EndFor
	\For{$i \gets 1 \mathrel{\mathbf{to}} m$}
		\State $b \gets R_{a_{k-1}[i]}[k]$;
		\State $C[b] \gets C[b] + 1$;
		\State $a_k[C[b]] \gets a_{k-1}[i]$;
		\If{$P[b] = 0$}
		    $d_k[C[b]] \gets k + 1$;
        \Else
            $\ d_k[C[b]] \gets \max\{d_{k-1}[\ell] \colon P[b]{<}\ell{\le}i\}$;
        \EndIf
		\State $P[b] \gets i$;
   \EndFor
\State ~
\State ~
\State ~
\State ~
\State ~
\State ~
\end{algorithmic}
\vskip1.6mm
\caption{The basic pBWT algorithm computing $a_k$~and~$d_k$ from $a_{k-1}$ and $d_{k-1}$.}\label{fig:pbwt:basic}
\end{subfigure}
\hfill
\begin{subfigure}[t]{.51\textwidth}
\begin{algorithmic}
\small
	\State zero initialize $C[0,|\Sigma|]$ and $P[0, |\Sigma| - 1]$
	\For{$i \gets 1 \mathrel{\mathbf{to}} m$}
        \State $C[R_i[k] + 1] \gets C[R_i[k] + 1] + 1$;
    \EndFor
	\For{$i \gets 1 \mathrel{\mathbf{to}} |\Sigma|-1$}
        $C[i] \gets C[i] + C[i - 1]$;
    \EndFor
	\For{$i \gets 1 \mathrel{\mathbf{to}} m$}
		\State $b \gets R_{a_{k-1}[i]}[k]$;
		\State $C[b] \gets C[b] + 1$;
		\State $a_k[C[b]] \gets a_{k-1}[i]$;
		\State $a_{k-1}[i] \gets i + 1$;\Comment{erase $a_{k-1}[i]$}
		\If{$P[b] = 0$}
		    $d_k[C[b]] \gets k + 1$;
        \Else
            $\ d_k[C[b]] \gets \mathsf{maxd}(P[b] + 1, i)$;
        \EndIf
		\State $P[b] \gets i$;
   \EndFor
\Function{$\mathsf{maxd}$}{$j$, $i$}
	\If{$j \ne i$}
		\State $d_{k-1}[j]{\gets}\max\{d_{k-1}[j], \mathsf{maxd}(a_{k-1}[j], i)\}$;
		\State $a_{k-1}[j] \gets i + 1$;
    \EndIf
    \State\Return{$d_{k-1}[j]$};
\EndFunction
\end{algorithmic}
\caption{The algorithm with simple RMQ; $a_{k-1}$ and $d_{k-1}$ are used as auxiliary arrays (and corrupted).}\label{fig:pbwt:modified}
\end{subfigure}
\caption{The computation of $a_k$ and $d_k$ from $a_{k-1}$ and $d_{k-1}$ in the pBWT.}
\end{figure}

Given $a_{k-1}$ and $d_{k-1}$, we are to show that the algorithm from Figure~\ref{fig:pbwt:basic} correctly computes $a_k$ and $d_k$. Since, for any $i, j \in [1,m]$, we have $R'_{i,k} \le R'_{j,k}$ iff either $R_i[k] < R_j[k]$, or $R_i[k] = R_j[k]$ and $R'_{i,k-1} \le R'_{j,k-1}$ lexicographically, it is easy to see that the array $a_k$ can be deduced from $a_{k-1}$ by radix sorting the sequence of pairs $\{ (R_{a_{k-1}[i]}[k], R'_{a_{k-1}[i],k-1}) \}_{i=1}^m$. Further, since, by definition of $a_{k-1}$, the second components of the pairs are already in a sorted order, it remains to sort the first components by the counting sort. Accordingly, in Figure~\ref{fig:pbwt:basic}, the first loop counts occurrences of letters in the sequence $\{R_i[k]\}_{i=1}^m$ using an auxiliary array $C[0,|\Sigma|]$; as is standard in the counting sort, the second loop modifies the array $C$ so that, for each letter $b \in [0,|\Sigma|{-}1]$, $C[b] + 1$  is the first index of the ``bucket'' that will contain all $a_{k-1}[i]$ such that $R_{a_{k-1}[i]}[k] = b$; finally, the third loop fills the buckets incrementing the indices $C[b] \gets C[b] + 1$, for $b = R_{a_{k-1}[i]}[k]$, and performing the assignments $a_k[C[b]] \gets a_{k-1}[i]$, for $i = 1,\ldots,m$. Thus, the array $a_k$ is computed correctly. All is done in $O(m + |\Sigma|)$ time, which is $O(m)$ since the input alphabet is $[0,|\Sigma|{-}1]$ and $|\Sigma| = O(m)$.

The last three lines of the algorithm are responsible for computing $d_k$. Denote the length of the longest common prefix of any strings $s_1$ and $s_2$ by $\lcp(s_1, s_2)$. The computation of $d_k$ relies on the following well-known fact: given a sequence of strings $s_1, \ldots, s_r$ such that $s_1 \le \cdots \le s_r$ lexicographically, one has $\lcp(s_1, s_r) = \min\{\lcp(s_{i-1}, s_{i}) \colon 1 < i \le r\}$. Suppose that the last loop of the algorithm, which iterates through all $i$ from $1$ to $m$, assigns $a_k[i'] \gets a_{k-1}[i]$, for a given $i \in [1,m]$ and some $i' = C[b]$. Let $j$ be the maximum integer such that $j < i$ and $R_{a_{k-1}[j]}[k] = R_{a_{k-1}[i]}[k]$ (if any). The definition of $a_k$ implies that $a_k[i' - 1] = a_{k-1}[j]$ if such $j$ exists. Hence, $\lcp(R'_{a_k[i' - 1],k}, R'_{a_k[i'],k}) = 1 + \min\{\lcp(R'_{a_{k-1}[\ell - 1], k-1}, R'_{a_{k-1}[\ell], k-1}) \colon j{<}\ell{\le}i\}$ if such number $j$ does exist, and $\lcp(R'_{a_k[i' - 1],k}, R'_{a_k[i'],k}) = 0$ otherwise. Therefore, since $d_k[i']$ equals $k + 1 - \lcp(R'_{a_k[i'],k}, R'_{a_k[i'-1],k})$, we have either $d_k[i'] = \max\{d_{k-1}[\ell] \colon j < \ell \le i\}$ or $d_k[i'] = k + 1$ according to whether the required $j$ exists. To find $j$, we simply maintain an auxiliary array $P[0..|\Sigma|{-}1]$ such that on the $i$th loop iteration, for any letter $b \in [0,|\Sigma|{-}1]$, $P[b]$ stores the position of the last seen $b$ in the sequence $R_{a_{k-1}[1]}[k], R_{a_{k-1}[2]}[k], \ldots, R_{a_{k-1}[i-1]}[k]$, or $P[b] = 0$ if $b$ occurs for the first time. Thus, $d_k$ is indeed computed correctly.

In order to calculate the maximums $\max\{d_{k-1}[\ell] \colon P[b] \le \ell \le i\}$ in $O(1)$ time, we build a range maximum query (RMQ) data structure on the array $d_{k-1}[1,m]$ in $O(m)$ time (e.g., see~\cite{FischerHeun}). Thus, the running time of the algorithm from Figure~\ref{fig:pbwt:basic} is, evidently, $O(m)$.

\begin{lemma}
The arrays $a_k$ and $d_k$ can be computed from $a_{k-1}$ and $d_{k-1}$ in $O(m)$ time.
\end{lemma}

In practice the bottleneck of the algorithm is the RMQ data structure, which, although answers queries in $O(1)$ time, has a sensible constant under the big-O in the construction time. We could naively compute the maximums by scanning the ranges $d_{k-1}[P[b]{+}1,i]$ from left to right but such algorithm works in quadratic time since same ranges of $d_{k-1}$ might be processed many times in the worst case. Our key idea is to store the work done by a simple scanning algorithm to reused it in future queries. We store this information right in the arrays $a_{k-1}$ and $d_{k-1}$ rewriting them; in particular, since $a_{k-1}$ is accessed sequentially from left to right in the last loop, the range $a_{k-1}[1,i]$ is free to use after the $i$th iteration.

More precisely, after the $i$th iteration of the last loop, the subarrays $a_{k-1}[1,i]$ and $d_{k-1}[1,i]$ are modified so that the following invariant holds: for any $j \in [1,i]$, $j < a_{k-1}[j] \le i + 1$ and $d_{k-1}[j] = \max\{d'_{k-1}[\ell] \colon j \le \ell < a_{k-1}[j]\}$, where $d'_{k-1}$ denotes the original array $d_{k-1}$ before modifications; note that the invariant holds if one simply puts $a_{k-1}[j] = j + 1$ without altering $d_{k-1}[j]$. Then, to compute $\max\{d'_{k-1}[\ell] \colon j\le \ell \le i\}$, we do not have to scan all elements but can ``jump'' through the chain $j, a_{k-1}[j], a_{k-1}[a_{k-1}[j]], \ldots, i$ and use maximums precomputed in $d_{k-1}[j], d_{k-1}[a_{k-1}[j]], d_{k-1}[a_{k-1}[a_{k-1}[j]]], \ldots, d_{k-1}[i]$; after this, we redirect the ``jump pointers'' in $a_{k-1}$ to $i + 1$ and update the maximums in $d_{k-1}$ accordingly. This idea is implemented in Figure~\ref{fig:pbwt:modified}. Notice the new line $a_{k-1}[i] \gets i + 1$ in the main loop (it is commented), which erases $a_{k-1}[i]$ and makes it a part of the ``jump table''. The correctness of the algorithm is clear. But it is not immediate even that the algorithm works in $O(m\log m)$ time. We prove the upper bound $O(m\log|\Sigma|)$ on the running time, which is a quite strong guarantee considering that in our problem the alphabet often is very small.

\begin{lemma}\label{lem:mlogsigma}
The algorithm from Figure~\ref{fig:pbwt:modified} computes the arrays $a_k$ and $d_k$ from $a_{k-1}$ and $d_{k-1}$ in $O(m\log|\Sigma|)$ time.
\end{lemma}
\begin{proof}
Fix $i \in [1,m]$. The $i$th iteration of the last loop in the algorithm computes the maximum in a range $d'_{k-1}[i',i]$, where $d'_{k-1}$ is the original array $d_{k-1}$ before modifications and $i' = P[b] + 1$ for some $b$ and $P$. Let $\ell_i = i - i'$. Denote $\tilde{\ell} = \frac{1}{m} \sum_{i=1}^m \ell_i$, the ``average query length''. We are to prove that the running time of the algorithm is $O(m\log\tilde{\ell})$, which implies the result since $m\tilde{\ell} = \sum_{i=1}^m \ell_i$ and, obviously, $\sum_{i=1}^m \ell_i \le \sigma m$.

We say that a position $j$ is \emph{touched} if the function $\mathsf{maxd}$ is called with its first argument equal to $j$. Clearly, it suffices to prove that the total number of touches is $O(m\log\tilde{\ell})$. While processing the query $\mathsf{maxd}(i{-}\ell_i, i)$, we may have touched many positions. Denote the sequence of all such position, for the given $i$, by $i_1, \ldots, i_r$; in other words, at the time of the query $\mathsf{maxd}(i{-}\ell_i, i)$, we have $i_1 = i - \ell_i$, $i_j = a_{k-1}[i_{j-1}]$ for $j \in [2,r]$, and $i_r = i$. Obviously, $i_1 < \cdots < i_r$. We say that, for $j \in [1,r{-}1]$, the touch of $i_j$ in the query $\mathsf{maxd}(i{-}\ell_i, i)$ is \emph{scaling} if there exists an integer $r$ such that $i - i_j > 2^r$ and $i - i_{j+1} \le 2^r$ (see Figure~\ref{fig:touches}). We count separately the total number of scaling and non-scaling touches in all $i$.

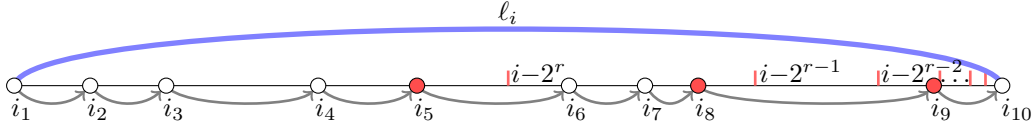
\begin{figure}[h]
\begin{tikzpicture}
\tikzstyle{every node}=[anchor=base,minimum width=0.7cm,inner sep=1pt]
\draw[black] (0,0) -- (13,0);
\draw[red!60, line width=1pt] (6.5,0) -- (6.5,0.2);
\draw[red!60, line width=1pt] (9.75,0) -- (9.75,0.2);
\draw[red!60, line width=1pt] (11.37,0) -- (11.37,0.2);
\draw[red!60, line width=1pt] (12.18,0) -- (12.18,0.2);
\draw[red!60, line width=1pt] (12.58,0) -- (12.58,0.2);
\draw[red!60, line width=1pt] (12.78,0) -- (12.78,0.2);
\draw[blue!50, line width=2pt] (0,0) .. controls (0.2,1) and (12.8,1) .. node[black,above]{$\ell_i$} (13,0);
\node at (6.9,0.05) () {$i{-}2^r$};
\node at (10.35,0.05) () {$i{-}2^{r-1}$};
\node at (11.95,0.05) () {$i{-}2^{r-2}$};
\node at (12.38,0.05) () {$\ldots$};
\draw[->,black!50, line width=1pt] (0,0) .. controls (0.2,-0.3) and (0.7,-0.3) .. (0.93,-0.07);
\draw[->,black!50, line width=1pt] (1,0) .. controls (1.2,-0.3) and (1.7,-0.3) .. (1.93,-0.07);
\draw[->,black!50, line width=1pt] (2,0) .. controls (2.2,-0.3) and (3.7,-0.3) .. (3.93,-0.07);
\draw[->,black!50, line width=1pt] (4,0) .. controls (4.2,-0.3) and (5.0,-0.3) .. (5.23,-0.07);
\draw[->,black!50, line width=1pt] (5.3,0) .. controls (5.5,-0.3) and (7.0,-0.3) .. (7.23,-0.07);
\draw[->,black!50, line width=1pt] (7.3,0) .. controls (7.5,-0.3) and (8.0,-0.3) .. (8.23,-0.07);
\draw[->,black!50, line width=1pt] (8.3,0) .. controls (8.5,-0.3) and (8.7,-0.3) .. (8.93,-0.07);
\draw[->,black!50, line width=1pt] (9,0) .. controls (9.2,-0.3) and (11.8,-0.3) .. (12.03,-0.07);
\draw[->,black!50, line width=1pt] (12.1,0) .. controls (12.3,-0.3) and (12.7,-0.3) .. (12.93,-0.07);
\node at (0.1,-0.4) () {$i_1$};
\node at (1.1,-0.4) () {$i_2$};
\node at (2.1,-0.4) () {$i_3$};
\node at (4.1,-0.4) () {$i_4$};
\node at (5.4,-0.4) () {$i_5$};
\node at (7.4,-0.4) () {$i_6$};
\node at (8.4,-0.4) () {$i_7$};
\node at (9.1,-0.4) () {$i_8$};
\node at (12.2,-0.4) () {$i_9$};
\node at (13.2,-0.4) () {$i_{10}$};
\filldraw[fill=white] (0,0) circle [radius=0.1];
\filldraw[fill=white] (1,0) circle [radius=0.1];
\filldraw[fill=white] (2,0) circle [radius=0.1];
\filldraw[fill=white] (4,0) circle [radius=0.1];
\filldraw[fill=red!70] (5.3,0) circle [radius=0.1];
\filldraw[fill=white] (7.3,0) circle [radius=0.1];
\filldraw[fill=white] (8.3,0) circle [radius=0.1];
\filldraw[fill=red!70] (9,0) circle [radius=0.1];
\filldraw[fill=red!70] (12.1,0) circle [radius=0.1];
\filldraw[fill=white] (13,0) circle [radius=0.1];
\end{tikzpicture}
\caption{RMQ query on a range $[i - \ell_i, i]$; scaling touches are red.}\label{fig:touches}
\end{figure}

For position $j$, denote by $p(j)$ the number of non-scaling touches of $j$. We are to prove that $P = \sum_{j=1}^m p(j) \le 2 m \log\tilde{\ell}$. Let $q_h(j)$ denote the value of $a_{k-1}[j] - j$ in the $h$th non-scaling touch of $j$, for $h \in [1,p(j)]$. Suppose that this $h$th touch happens  during the processing of a query $\mathsf{maxd}(i - \ell_i, i)$. By the definition, $j + q_h(j)$ follows $j$ in the sequence of touched positions. Since the touch of $j$ is non-scaling, we have $i - j > j + q_h(j) > 2^r$, where $r$ is the largest integer such that $i - j > 2^r$, and hence, $q_h(j) < 2^r$. Since $\mathsf{maxd}(i - \ell_i, i)$ assigns $a_{k-1}[j] \gets i + 1$, we have $a_{k-1}[j] - j > i - j > 2^r$ after the query. In other words, we had $a_{k-1}[j] - j = q_h(j) < 2^r$ before the query and have $a_{k-1}[j] - j > 2^r$ after. This immediately implies that $q_h(j) \ge 2^{h-1}$, for $h \in [1,p(j)]$, and, therefore, every position can be touched in the non-scaling way at most $O(\log m)$ times, implying $P = O(m\log m)$. But we can deduce a stronger bound. Since the sum of all values $j - a_{k-1}[j]$ for all positions $j$ touched in a query $\mathsf{maxd}(i - \ell_i, i)$ is equal to $\ell_i$, it is obvious that $\sum_{j = 1}^m \sum_{h = 1}^{p(j)} q_h(j) \le \sum_{i=1}^m \ell_i = m\tilde{\ell}$. On the other hand, we have $\sum_{j = 1}^m \sum_{h = 1}^{p(j)} q_h(j) \ge \sum_{j = 1}^m \sum_{h = 1}^{p(j)} 2^{h-1} = \sum_{j=1}^m 2^{p(j)} - m$. The well-known property of the convexity of the exponent is that the sum $\sum_{j=1}^m 2^{p(j)}$ is minimized whenever all $p(j)$ are equal and maximal, i.e., $\sum_{j=1}^m 2^{p(j)} \ge \sum_{j=1}^m 2^{P / m}$. Hence, once $P > 2 m \log\tilde{\ell}$, we obtain $\sum_{j = 1}^m \sum_{h = 1}^{p(j)} q_h(j) \ge  \sum_{j = 1}^m 2^{P / m} - m > m\tilde{\ell}^2 - m$, which is larger than $m\tilde{\ell}$ for $\tilde{\ell} \ge 2$ (the case $\tilde{\ell} < 2$ is trivial), contradicting $\sum_{j = 1}^m \sum_{h = 1}^{p(j)} q_h(j) \le m\tilde{\ell}$. Thus, $P = \sum_{j=1}^m p(j) \le 2 m \log\tilde{\ell}$.

It remains to consider scaling touches. The definition implies that each query $\mathsf{maxd}(i{-}\ell_i, i)$ performs at most $\log\ell_i$ scaling touches. Thus, it suffices to upperbound $\sum_{i=1}^m \log\ell_i$. Since the function $\log$ is concave, the sum $\sum_{i=1}^m \log\ell_i$ is maximized whenever all $\ell_i$ are equal and maximal, i.e., $\sum_{i=1}^m \log\ell_i \le \sum_{i=1}^m \log(\frac{1}m \sum_{j=1}^m \ell_j) = m\log\tilde{\ell}$, hence the result follows.
\end{proof}

\subsection{Modification of the pBWT}

We are to modify the basic pBWT construction algorithm in order to compute the sequence $j_{k,1}, \ldots, j_{k,r_k}$ of all positions $j \in [1,k-L]$ in which $\card{\mathcal{R}[j,k]} \ne \card{\mathcal{R}[j + 1,k]}$, and to calculate the numbers $\card{\mathcal{R}[j_{k,h+1},k]}$ and $\min\{ M(j) \colon j_{k,h} \le j < j_{k,h+1} \}$, for $h \in [0,r_k]$ (assuming $j_{k,0} = 0$ and $j_{k,r_k+1} = k - L + 1$); see the beginning of the section. As it follows from~\eqref{eq:mainrelation}, these numbers are sufficient to calculate $M(k)$, as defined in~\eqref{eq:ukkonen} and~\eqref{eq:mainrelation}, in $O(m)$ time. The following lemma reveals relations between the sequence $j_{k,1}, \ldots, j_{k,r_k}$ and the array $d_k$.

\begin{lemma}\label{lem:dchanges}
Consider recombinants $\mathcal{R} = \{R_1, \ldots, R_m\}$, each having length $n$. For $k \in [1,n]$ and $j \in [1,k - 1]$, one has $\card{\mathcal{R}[j,k]} \ne \card{\mathcal{R}[j + 1,k]}$ iff $j = d_k[i] - 1$ for some $i \in [1,m]$.
\end{lemma}
\begin{proof}
Suppose that $\card{\mathcal{R}[j,k]} \ne \card{\mathcal{R}[j + 1,k]}$. It is easy to see that $\card{\mathcal{R}[j,k]} > \card{\mathcal{R}[j + 1,k]}$, which implies that there are two indices $h$ and $h'$ such that $R_h[j + 1, k] = R_{h'}[j + 1, k]$ and $R_h[j] \ne R_{h'}[j]$. Denote by $a_k^{-1}[h]$ the number $x$ such that $a_k[x] = h$. Without loss of generality, assume that $a_k^{-1}[h] < a_k^{-1}[h']$. Then, there exists $i \in [a_k^{-1}[h] + 1, a_k^{-1}[h']]$ such that $R_{a_k[i - 1]}[j + 1, k] = R_{a_k[i]}[j + 1, k]$ and $R_{a_k[i - 1]}[j] \ne R_{a_k[i]}[j]$. Hence, $d_k[i] = j + 1$.

Suppose now that $j \in [1, k - 1]$ and $j = d_k[i] - 1$, for some $i \in [1,m]$. Since $j < k$ and $d_k[1] = k + 1$, we have $i > 1$. Then, by definition of $d_k$, $R_{a_k[i-1]}[j + 1, k] = R_{a_k[i]}[j + 1, k]$ and $R_{a_k[i-1]}[j] \ne R_{a_k[i]}[j]$, i.e., $R_{a_k[i]}[j + 1, k]$ can be ``extended'' to the left in two different ways, thus producing two distinct strings in the set $\mathcal{R}[j, k]$. Therefore, $\card{\mathcal{R}[j, k]} > \card{\mathcal{R}[j + 1, k]}$.
\end{proof}

Denote by $r$ the number of distinct integers in the array $d_k$. Clearly, $r$ may vary from $1$ to $m$. For integer $\ell$, define $M'(\ell) = M(\ell)$ if $1 \le \ell \le k - L$, and $M'(\ell) = +\infty$ otherwise. Our modified algorithm does not store $d_k$ but stores the following four arrays (but we still often refer to $d_k$ for the sake of analysis):

\begin{itemize}
\item $s_k[1,r]$ contains all distinct elements from $d_k[1,m]$ in the increasing sorted order;
\item $e_k[1,m]$: for $j \in [1,r]$, $e_k[j]$ is equal to the unique index such that $s_k[e_k[j]] = d_k[j]$;
\item $t_k[1,r]$: for $j \in [1,r]$, $t_k[j]$ is equal to the number of times $s_k[j]$ occurs in $d_k[1,m]$;
\item $u_k[1,r]$: for $j \in [1,r]$, $u_k[j] = \min\{M'(\ell) \colon s_k[j{-}1]{-}1 \le \ell < s_k[j]{-}1\}$, assuming $s_k[0] = 1$.
\end{itemize}

\begin{example}
In Example~\ref{firstexample}, where $m = 6$, $k = 7$, and $\Sigma = \{a,c,t\}$, we have $r = 4$, $s_k = [3, 5, 7, 8]$, $t_k = [2, 1, 1, 2]$, $e_k = [4, 4, 2, 1, 3, 1]$. Further, suppose that $L = 3$, so that $k - L = 4$. Then, $u_k[1] = M(1)$, $u_k[2] = \min\{M(2), M(3)\}$, $u_k[3] = \min\{M(4), M'(5)\} = M(4)$ since $M'(5) = +\infty$, and $u_k[4] = M'(6) = +\infty$.
\end{example}

By the definition of $d_k$, we have $d_k[1] = k + 1$ and, hence, the last element of $s_k$, $s_k[|s_k|]$, must be equal to $k + 1$. Assume that $s_k[0] = 1$. Then, the array $s_k$ defines a splitting of the range $[0, k - 1]$ into the disjoint segments $[s_k[j - 1] - 1, s_k[j] - 2]$, for $j \in [1,|s_k|]$. Note that only the first segment might be empty and only if $s_k[1] = 1$. Recall that $0 = j_{k,0} < j_{k,1} < \cdots j_{k,r_k} < j_{k,r_k + 1} = k - L + 1$. It follows from Lemma~\ref{lem:dchanges} that the first $r_k$ non-empty segments $[s_k[j - 1] - 1, s_k[j] - 2]$ correspond to the segments $[j_{k,h-1}, j_{k,h} - 1]$, for $h \in [1,r_k]$, and the $(r_k + 1)$st non-empty segment $[s_k[j - 1] - 1, s_k[j] - 2]$ covers the point $k - L$, so that $[j_{k,r_k}, j_{k,r_k+1} - 1]$ is a prefix of this segment. It is clear that $u_k[j] \ne +\infty$ only if the segment $[s_k[j - 1] - 1, s_k[j] - 2]$ intersects the range $[1, k - L]$ and, thus, corresponds to a segment $[j_{k,h-1}, j_{k,h} - 1]$, for $h \in [1,r_k + 1]$, in the above sense. Therefore, since $M'(\ell) = +\infty$ for $\ell < 1$ and $\ell > k - L$ and, thus, such values $M'(\ell)$ do not affect, in a sense, the minima stored in $u_k$, one can rewrite~\eqref{eq:mainrelation} as follows:
\begin{equation}\label{eq:computerelation}
M(k) =
\begin{cases}
+\infty & \text{ if }k < L,\\
\card{\mathcal{R}[1,k]} & \text{ if } L \le k < 2L,\\
\min\limits_{1 \le j \le |u_k|} \max\{\card{\mathcal{R}[s_k[j] - 1, k]}, u_k[j]\} & \text{ if } k \ge 2L.
\end{cases}
\end{equation}

It remains to compute the numbers $\card{\mathcal{R}[s_k[j] - 1, k]}$, for $j \in [1,|s_k|]$.

\begin{lemma}\label{lem:rcompute}
Consider a set of recombinants $\mathcal{R} = \{R_1, \ldots, R_m\}$, each of which has length~$n$. For $k \in [1,n]$ and $j \in [1,|s_k|]$, one has $\card{\mathcal{R}[s_k[j] - 1, k]} = t_k[j] + t_k[j + 1] + \cdots + t_k[|t_k|]$.
\end{lemma}
\begin{proof}
Denote $\ell = k - s_k[j] + 1$, so that $\mathcal{R}[s_k[j] - 1, k] = \mathcal{R}[k - \ell, k]$. Suppose that $\ell = 0$. Note that $R_{a_k[1]}[k] \le \cdots \le R_{a_k[m]}[k]$. Since $d_k[i] = k + 1$ iff either $i = 1$ or $R_{a_k[i-1]}[k] \ne R_{a_k[i]}[k]$, it is easy to see that $\card{\mathcal{R}[k,k]}$, the number of distinct letters $R_i[k]$, is equal to the number of time $k + 1 = s_k[|s_k|]$ occurs in $d_k$, i.e., $t_k[|t_k|]$.

Suppose that $\ell > 0$. It suffices to show that $\card{\mathcal{R}[k - \ell, k]} - \card{\mathcal{R}[k - \ell + 1, k]} = t_k[j]$. For $i \in [1,m]$, denote by $R'_i$ the string $R_i[k] R_i[k - 1] \cdots R_i[k - \ell]$. Fix $w \in \mathcal{R}[k - \ell + 1, k]$. Since $R'_{a_k[1]} \le \cdots \le R'_{a_k[m]}$ lexicographically, there are numbers $h$ and $h'$ such that $R_{a_k[i]}[k - \ell + 1, k] = w$ iff $i \in [h,h']$. Further, we have $R_{a_k[h]}[k - \ell] \le R_{a_k[h + 1]}[k - \ell] \le \cdots \le R_{a_k[h']}[k - \ell]$. Thus, by definition of $d_k$, for $i \in [h + 1, h']$, we have $R_{a_k[i-1]}[k - \ell] \ne R_{a_k[i]}[k - \ell]$ iff $d_k[i] = k - \ell + 1 = s_k[j]$. Note that $d_k[h] > s_k[j]$. Therefore, the number of strings $R_i[k - \ell, k]$ from $\mathcal{R}[k - \ell, k]$ having suffix $w$ is equal to one plus the number of integers $s_k[j]$ in the range $d_k[h, h']$, which implies $\card{\mathcal{R}[k - \ell, k]} - \card{\mathcal{R}[k - \ell + 1, k]} = t_k[j]$.
\end{proof}

In particular, it follows from Lemmas~\ref{lem:dchanges} and~\ref{lem:rcompute} that $\card{\mathcal{R}[1,k]} = t_k[1] + \cdots + t_k[|t_k|]$. Thus, by~\eqref{eq:computerelation}, one can simply calculate $M(k)$ in $O(m)$ time using the arrays $t_k$ and $u_k$.

The arrays $e_k$, $s_k$, $t_k$, $u_k$ along with $a_k$ are computed from $e_{k-1}, s_{k-1}, t_{k-1}, u_{k-1}, a_{k-1}$ by Algorithm~\ref{fig:algorithm}. Let us analyze this algorithm.

\begin{algorithm}[h]
\begin{algorithmic}[1]
\small
    \State copy $s_{k-1}$ into $s_k$ and add the element $k + 1$ to the end of $s_k$, thus incrementing $\card{s_k}$;
    \State copy $u_{k-1}$ into $u_k$ and add the element $M'(k - 1)$ to the end of $u_k$, thus incrementing $\card{u_k}$;\label{lst:prepareu}
	\State zero initialize $C[0,|\Sigma|]$, $P[0, |\Sigma| - 1]$, and $t_k[1,\card{s_k}]$;\label{lst:beginbasic}
	\For{$i \gets 1 \mathrel{\mathbf{to}} m$}
        $C[R_i[k] + 1] \gets C[R_i[k] + 1] + 1$;
    \EndFor
	\For{$i \gets 1 \mathrel{\mathbf{to}} |\Sigma|-1$}
        $C[i] \gets C[i] + C[i - 1]$;
    \EndFor
	\For{$i \gets 1 \mathrel{\mathbf{to}} m$}
		\State $b \gets R_{a_{k-1}[i]}[k]$;
		\State $C[b] \gets C[b] + 1$;
		\State $a_k[C[b]] \gets a_{k-1}[i]$;
		\If{$P[b] = 0$}
		    $e_k[C[b]] \gets |s_k|$;
        \Else
            $\ e_k[C[b]] \gets \max\{e_{k-1}[\ell] \colon P[b] < \ell \le i\}$ 
        \EndIf
		\State $P[b] \gets i$;
    \EndFor\label{lst:endbasic}
    \For{$i \gets 1 \mathrel{\mathbf{to}} m$}\label{lst:fillt}
        $t_k[e_k[i]] \gets t_k[e_k[i]] + 1;$
    \EndFor
    \State $j \gets 1$;\label{lst:initj}
    \State add a new ``dummy'' element $+\infty$ to the end of $u_k$ and $u_{k-1}$;\label{lst:dummyadd}
	\For{$i \gets 1 \mathrel{\mathbf{to}} \card{s_k}$}\label{lst:fixloopbeg}
        \State $u_k[j] \gets \min\{u_k[j], u_{k-1}[i]\}$;\label{lst:correctu}
        \If{$t_k[i] \ne 0$}\label{lst:ifnonzerobeg}
            \State $tmp[i] \gets j;$
            \State $s_k[j] \gets s_k[i],\;t_k[j] \gets t_k[i],\;u_k[j + 1] \gets u_{k-1}[i + 1];$\label{lst:fixstu}
            \If{$s_k[j] - 1 > k - L \mathrel{\mathbf{and}} (j = 1 \mathrel{\mathbf{or}} s_k[j - 1] - 1 \le k - L)$}\label{lst:correctuif}
                $u_k[j] \gets \min\{u_k[j], M(k - L)\};$
            \EndIf
            \State $j \gets j + 1;$
        \EndIf\label{lst:ifnonzeroend}
    \EndFor\label{lst:fixloopend}
    \State shrink $s_k$, $t_k$, and $u_k$ to $j - 1$ elements, so that $\card{s_k} = \card{t_k} = \card{u_k} = j - 1$;\label{lst:fixlength}
    \For{$i \gets 1 \mathrel{\mathbf{to}} m$}\label{lst:fixe}
        $e_k[i] \gets tmp[e_k[i]];$
    \EndFor
\end{algorithmic}
\caption{The algorithm computing $e_k$, $s_k$, $t_k$, $u_k$, $a_k$.}\label{fig:algorithm}
\end{algorithm}

By definition, $d_{k-1}[i] = s_{k-1}[e_{k-1}[i]]$ for $i \in [1,m]$. The first line of the algorithm initializes $s_k$ so that $d_{k-1}[i] = s_k[e_{k-1}[i]]$, for $i \in [1,m]$, and $s_k[|s_k|] = k + 1$. Since after this initialization $s_k$, obviously, is in the sorted order, one has, for $i, j \in [1,m]$, $e_{k-1}[i] \le e_{k-1}[j]$ iff $d_{k-1}[i] \le d_{k-1}[j]$ and, therefore, for $\ell \in [i,j]$, one has $d_{k-1}[\ell] = \max\{d_{k-1}[\ell'] \colon i \le \ell' \le j\}$ iff $e_{k-1}[\ell] = \max\{e_{k-1}[\ell'] \colon i \le \ell' \le j\}$. Based on this observation, we fill $e_k$ in lines~\ref{lst:beginbasic}--\ref{lst:endbasic} so that $d_k[i] = s_k[e_k[i]]$, for $i \in [1,m]$, using exactly the same algorithm as in Figure~\ref{fig:pbwt:basic}, where $d_k$ is computed, but instead of the assignment $d_k[C[b]] \gets k + 1$, we have $e_k[C[b]] \gets |s_k|$ since $s_k[|s_k|] = k + 1$. Here we also compute $a_k$ in the same way as in Figure~\ref{fig:pbwt:basic}.

The loop in line~\ref{lst:fillt} fills $t_k$ so that, for $i \in [1,|s_k|]$, $t_k[i]$ is the number of occurrences of the integer $i$ in $e_k$ ($t_k$ was zero initialized in line~\ref{lst:beginbasic}). Since, for $i \in [1,m]$, we have $d_k[i] = s_k[e_k[i]]$ at this point, $t_k[i]$ is also the number of occurrences of the integer $s_k[i]$ in $d_k[1,m]$.

By definition, $s_k$ must contain only elements from $d_k$, but this is not necessarily the case in line~\ref{lst:initj}. In order to fix $s_k$ and $t_k$, we simply have to remove all elements $s_k[i]$ for which $t_k[i] = 0$, moving all remaining elements of $s_k$ and non-zero elements of $t_k$ to the left accordingly. Suppose that, for some $h$ and $i$, we have $e_k[h] = i$ and the number $s_k[i]$ is moved to $s_k[j]$, for some $j < i$, as we fix $s_k$. Then, $e_k[h]$ must become~$j$. We utilize an additional temporary array $tmp[1,|s_k|]$ to fix $e_k$. The loop in lines~\ref{lst:fixloopbeg}--\ref{lst:fixloopend} fixes $s_k$ and $t_k$ in an obvious way; once $s_k[i]$ is moved to $s_k[j]$ during this process, we assign $tmp[i] = j$. Then, $s_k$, $t_k$, $u_k$ ($u_k$ is discussed below) are resized in line~\ref{lst:fixlength}, and the loop in line~\ref{lst:fixe} fixes $e_k$ using $tmp$.

Recall that $[s_k[j - 1] - 1, s_k[j] - 2]$, for $j \in [1,|s_k|]$, is a system of disjoint segments covering $[0, k - 1]$ (assuming $s_k[0] = 1$). It is now easy to see that this system is obtained from the system $[s_{k-1}[j - 1] - 1, s_{k-1}[j] - 2]$, with $j \in [1, |s_{k-1}|]$ (assuming $s_{k-1}[0] = 1$), by adding the new segment $[k - 1, k - 1]$ and joining some segments together. The second line of the algorithm copies $u_{k-1}$ into $u_k$ and adds $M'(k - 1)$ to the end of $u_k$, so that, for $j \in [1, |u_{k-1}|]$, $u_k[j]$ is equal to the minimum of $M'(\ell)$ for all $\ell$ from the segment $[s_{k-1}[j - 1] - 1, s_{k-1}[j] - 2]$ and $u_k[|u_{k-1}|{+}1] = M'(k - 1)$ is the minimum in the segment $[k - 1, k - 1]$. (This is not completely correct since $M'$ has changed as $k$ was increased; namely, $M'(k - L)$ was equal to $+\infty$ but now is equal to $M(k - L)$.) As we join segments removing some elements from $s_k$ in the loop~\ref{lst:fixloopbeg}--\ref{lst:fixloopend}, the array $u_k$ must be fixed accordingly: if $[s_k[j - 1] - 1, s_k[j] - 2]$ is obtained by joining $[s_{k-1}[h - 1] - 1, s_{k-1}[h] - 2]$, for $j' \le h \le j''$, then $u_k[j] = \min\{u_{k-1}[h] \colon j' \le h \le j''\}$. We perform such fixes in line~\ref{lst:correctu}, accumulating the latter minimum. We start accumulating a new minimum in line~\ref{lst:fixstu}, assigning $u_k[j + 1] \gets u_{k-1}[i + 1]$. If at this point the ready minimum accumulated in $u_k[j]$ corresponds to a segment containing the position $k - L$, we have to fix $u_k$ taking into account the new value $M'(k - L) = M(k - L)$; we do this in line~\ref{lst:correctuif}. To avoid accessing out of range elements in $u_k$ and $u_{k-1}$ in line~\ref{lst:fixstu}, we add a ``dummy'' element in, respectively, $u_k$ and $u_{k-1}$ in line~\ref{lst:dummyadd}.

Besides all the arrays of length $m$, the algorithm also requires access to $M(k - L)$ and, possibly, to $M(k - 1)$. During the computation of $M(k)$ for $k \in [1,n]$, we maintain the last $L$ calculated numbers $M(k - 1), M(k - 2), \ldots, M(k - L)$ in a circular array, so that the overall required space is $O(m + L)$; when $k$ is incremented, the array is modified in $O(1)$ time in an obvious way. Thus, we have proved the following result, implying Theorem~\ref{thm:maintheorem}.
\begin{lemma}
The arrays $a_k, e_k, s_k, t_k, u_k$ can be computed from $a_{k-1}, e_{k-1}, s_{k-1}, t_{k-1}, u_{k-1}$ and from the numbers $M(k - L)$ and $M(k - 1)$ in $O(m)$ time.
\end{lemma}

If, as in our case, one does not need $s_k, t_k, u_k$ for all $k$, the arrays $s_k$, $t_k$, $u_k$ can be modified in-place, i.e., $s_k$, $t_k$, $u_k$ can be considered as aliases for $s_{k-1}$, $t_{k-1}$, $u_{k-1}$, and yet the algorithm remains correct. Thus, we really need only 7 arrays in total: $a_k$, $a_{k-1}$, $e_k$, $e_{k-1}$, $s$, $t$, $u$, where $s$, $t$, $u$ serve as $s_k$, $t_k$, $u_k$ and the array $tmp$ can be organized in place of $a_{k-1}$ or $e_{k-1}$. It is easy to maintain along with each value $u_k[j]$ a corresponding position $\ell$ such that $u_k[j] = M'(\ell)$; these positions can be used then to restore the found segmentation of $\mathcal{R}$ using backtracking (see the beginning of the section). To compute $e_k$, instead of using an RMQ data structure, one can adapt in an obvious way the algorithm from Figure~\ref{fig:pbwt:modified} rewriting the arrays $a_{k-1}$ and $e_{k-1}$ during the computation, which is faster in practice but theoretically takes $O(m\log\sigma)$ time by Lemma~\ref{lem:mlogsigma}. We do not discuss further details as they are straightforward.

\end{document}